\documentclass[draftclsnofoot]{IEEEtran}
 \onecolumn
\usepackage{latexsym}
\usepackage{cite}
\usepackage{color}
\usepackage{bm}
\usepackage{amsmath, amssymb}
\usepackage[dvips]{graphics}
\usepackage{graphicx}
\usepackage{psfrag}
 \allowdisplaybreaks

\newtheorem{definition}{Definition}

\newtheorem{theorem}{Theorem}
\newtheorem{lemma}[theorem]{Lemma}

\newtheorem{example}{Example}

\begin{document}

\newcommand{\be}{\begin{equation}}
\newcommand{\ee}{\end{equation}}
\newcommand{\bea}{\begin{eqnarray}}
\newcommand{\eea}{\end{eqnarray}}
\newcommand{\beaa}{\begin{eqnarray*}}
\newcommand{\eeaa}{\end{eqnarray*}}

\title{Multiple Access Channel with Partial and Controlled Cribbing Encoders}

\author{Haim Permuter and  Himanshu Asnani \\
\thanks{H. Permuter is with the department of Electrical and Computer Engineering,
 Ben-Gurion University of the Negev, Beer-Sheva, Israel (haimp@bgu.ac.il).
H. Asnani is with the  Department of Electrical Engineering, Stanford University, CA, USA
(asnani@stanford.edu). }
}

 \maketitle \vspace{-1.4cm}

\begin{abstract}
In this paper we consider a multiple access channel (MAC) with partial cribbing encoders. This
means that each of two encoders obtains a deterministic function of the other encoder output with
or without delay. The partial cribbing scheme is especially motivated by the additive noise
Gaussian MAC since perfect cribbing  results in the degenerated case of full cooperation between
the encoders and requires an infinite entropy link. We derive a single letter characterization of
the capacity of the MAC with partial cribbing for the cases of causal and strictly causal
partial cribbing. Several numerical examples, such as quantized cribbing, are presented. We
further consider and derive the capacity region where the cribbing depends on actions that are
functions of the previous cribbed observations. In particular, we consider a scenario where the
action is ``to crib or not to crib" and show that a naive time-sharing strategy is not optimal.
\end{abstract}

\vspace{-0.0cm}
\section{Introduction}

In his remarkable dissertation \cite{willems82_dissertation}, Willems  introduced a new problem
of the multiple access channel (MAC) with cribbing encoders and derived its capacity region using
a novel decoding technique called ``backward decoding". Cribbing encoder refers to the case where
the encoder knows perfectly the other output encoder, possibly with delay or lookahead.
\begin{figure}[h!]{
\psfrag{B}[][][1]{Encoder 1} \psfrag{D}[][][1]{Encoder 2} \psfrag{m1}[][][1]{$m_1 \ \ \ \ $}
\psfrag{m2}[][][1]{$\in\{1,...,2^{nR_1}\}\ \ \ \ \ \ \ $} \psfrag{m3}[][][1]{$m_2 \ \ \ \ $}
\psfrag{m4}[][][1]{$\in\{1,...,2^{nR_2}\}\ \ \ \ \ \ \ $} \psfrag{P}[][][1]{$\ \ P_{Y|X_1,X_2}$}
\psfrag{x1}[][][1]{$\ \ \ \ \ \ \ \ \ \ \ \ X_{1,i}(m_1,Z_2^{i-1})$} \psfrag{x2}[][][1]{\; \; \;
\; \; $X_{2,i}(m_2,Z_1^{i})$} \psfrag{M}[][][1]{MAC} \psfrag{s}[][][1]{$S$}
\psfrag{Yi}[][][1]{$Y$} \psfrag{W}[][][1]{Decoder} \psfrag{t}[][][1]{$$}
\psfrag{G1}[][][0.9]{$Z_{1,i}=g_1(X_{1,i})$} \psfrag{G2}[][][0.9]{$Z_{2,i}=g_2(X_{2,i})$}

\psfrag{d}[][][1]{\ delay} \psfrag{a}[][][1]{a} \psfrag{b}[][][1]{b} \psfrag{c12}[][][1]{$
m_{12}\in$} \psfrag{c12b}[][][1]{$ \ \ \ \ \ \ \{1,...,2^{nC_{12}}\}$}

\psfrag{Y}[][][1]{$\ \ \ \ \ \ \hat m_1(Y^n)$} \psfrag{Y2}[][][1]{$\ \ \ \ \hat m_2(Y^n)$}

\centerline{\includegraphics[width=14cm]{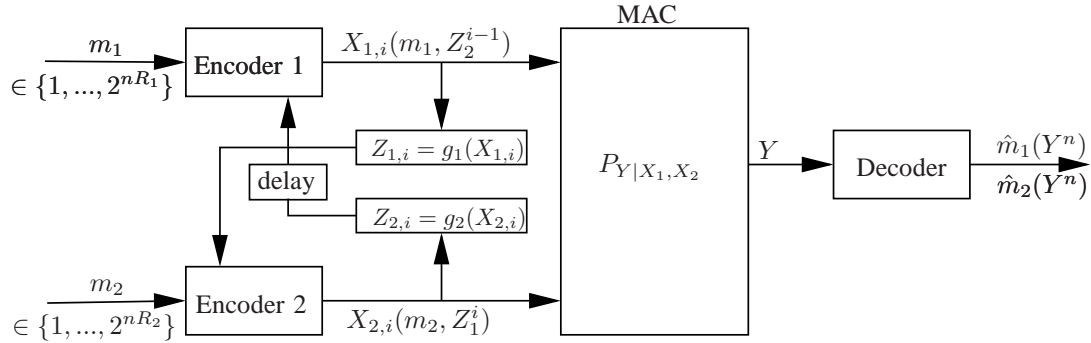}}

\caption{Partial (deterministic-function) cribbing. Each encoder observes a deterministic
function of the other encoder with or without delay. Encoder 1 observes the cribbing in a
strictly causal way, i.e., with delay,  and Encoder 2 observes the cribbing causally, i.e.,
without delay. The setting corresponds to Case B in this paper. }\label{f_mac_crib2}
 }\end{figure}
The work by Willems on MACs with cribbing encoders has been extended to the interference channel
\cite{BrossSteinberg10intereferenceCribbing}, and to state-dependent MAC
\cite{BrossLapidoth10MAC_cribbing_with_state}.
   However, for the Gaussian case, where the
encoder output is of a continuous alphabet, the cribbing idea is not an interesting case
\cite{Willems06cribbingRevisited} since it implies a full cooperation between the encoders
regardless of the delay of the cribbing. This is due to the fact that in a single epoch time a
noiseless continuous signal may transmit an infinite amount of information. Motivated by this
fact, we introduce in this paper {``\it partial cribbing"}, where one encoder only knows a
quantized version, or, more generally, a deterministic function of the coded output of other
encoder.

 In this paper we consider two kinds of partial cribbing: causal and strictly-causal. Causal
 partial cribbing means that at time $i$ the encoder observes (and uses) the partial cribbing signal without delay, i.e., $Z_i$.
 Strictly-causal
 partial cribbing means that at time $i$ the encoder observes the partial cribbing with a delay,
 i.e., $Z_{i-1}$. We derive the capacity region for two different cases according to the causality or the
strictly causality of the cribbing
\begin{itemize}
\item[]Case A: The cribbing for both encoders is {\it strictly-causal}.
\item[]Case B: The cribbing for one encoder is {\it causal} and for the other encoder is {\it strictly-causal}.
\end{itemize}
The setting that is depicted in Fig. \ref{f_mac_crib2} is the case where one encoder has causal
partial cribbing and the other strictly causal partial cribbing, namely Case B. To some extent,
the partial cribbing problem is related to the semi-deterministic relay channel
\cite{ElGamalrefA82SemiDetereministicRelay} which was solved using the partial decode and forward
technique \cite{CoverGamal79Relay}. The partial cribbing setting has a similar structure to the
semi-deterministic relay channel where Encoder 2 plays the role of relay and receives a
deterministic function of the output of Encoder 1. However, the MAC with partial cribbing  is
different from the semi-deterministic relay in the sense that Encoder 2 has its own message to
transmit in addition to its role of relaying information from Encoder 1. Another related problem
is the semi-deterministic broadcast channel\cite{GelfandPinsker80SemiDeterministicBroadcast},
where one of the receivers obtains a deterministic function of the input channel. In our problem
Encoder 1 ``is broadcasting" to Encoder 2 and to the decoder hence this part of the communication
resembles the semi-deterministic broadcast channel. However, in our problem of partial cribbing
only the decoder is actually required to decode the message error-free.

The coding scheme presented here for the partial cribbing uses the same techniques that were used
for the perfect cribbing, i.e., block Markov coding, Shannon's strategies, super position coding,
and backward decoding, and in addition to that, we use rate splitting in the code design. Rate
splitting is needed since Encoder 2 can decode only part of the message transmitted by Encoder 1.

Recently, several problems on ``action" in information theory have been considered in
\cite{Weissman10_action_dependet_state_channel,WeissmanPermuter09_VendingMachine,AsnaniPermuterWeissman10_probing_capacity,AsnaniPermuterWeissman10_to_feed_or_not}.
In these problems the side information is not freely available but depends on an action that has
a cost. The solution of partial cribbing allows us to consider the case where the cribbing is
action dependent. Namely, there is an action that is a function  of the previously cribbed
observations and this action  determines the current cribbing function. This kind of questions
may be raised in cognitive communication systems where sensing other users� signals is a resource
with a cost. In particular we show through a simple example where the action is ``to crib or not
to crib" that a time-sharing action is not necessarily optimal.

The remainder of the paper is organized as follows: In Section \ref{s_problem_def_main_res}, we
introduce the setting of MAC with partial cribbing and state the capacity region for
strictly-causal (Case A), as well as mixed causal and strictly-causal (Case B). In sections
\ref{s_converse} and \ref{s_achievability} respectively we provide the converse and achievability
proofs of the capacity region for each case of partial cribbing. In Section
\ref{s_commom_message} we consider the case where a common message, known to the encoders, needs
to be transmitted to the decoder in addition to the private messages. We show that no additional
auxiliary random variable is needed to characterize the capacity region since the partial
cribbing is utilized via generating a common message between the users. In Section
\ref{s_semi_deterministic} we consider the case where one of the encoders has no message to send;
hence it becomes a special case of the semi-deterministic relay channel with and without delay.
We show that indeed the region obtained via partial cribbing and the region obtained via a
semi-deterministic relay channel coincide. In Section \ref{s_Gaussian_mac} we consider a Gaussian
MAC with quantized cribbing. We provide a simple achievable scheme and show numerically that even
with a few bit quantizer we obtain an achievable region that is very close to the perfect
cribbing capacity region. In Section \ref{s_controlled_cribbing} we consider a scenario where a
limited-resource action controls the cribbing. In particular, we investigate an example where the
action is ``to crib or not to crib" and solve it analytically. In Section \ref{s_conclusion} we
conclude the paper and suggest some research directions that have not been yet solved such as
noncausal partial cribbing, noisy cribbing and a few action related problems.

\section{Problem definition and Main Result \label{s_problem_def_main_res}}

The MAC setting consists of two transmitters (encoders) and one
receiver (decoder). Each transmitter $l\in\{1,2\}$ chooses an index $m_l$
uniformly from the set $\{1,...,2^{nR_l}\}$ and independently of the
other transmitter. The input to the channel from encoder $l\in\{1,2\}$ is
denoted by $\{X_{l,1},X_{l,2},X_{l,3},...\}$. Encoder 1 and Encoder 2 obtain a deterministic function of the output of Encoder 2 and 1, respectively, of the form  $Z_{2,i}=g_2(X_{2,i}),$ and  $Z_{1,i}=g_1(X_{1,i}).$ The output of the
channel is denoted by $\{Y_1,Y_2,Y_3,...\}$.  The channel is characterized by a conditional
probability $P(y_i|x_{1,i},x_{2,i})$. The channel probability does not depend on the time
index $i$ and is memoryless, i.e.,
\begin{equation}\label{e_memoryless_def}
P(y_i|x_{1}^{i},x_{2}^{i},y^{i-1})=P(y_i|x_{1,i},x_{2,i}),
\end{equation}
where the superscripts denote sequences in the following way:
$x_l^i=(x_{l,1},x_{l,2},...,x_{l,i}), \; l\in\{1,2\}$.  Since the settings in this paper do not
include feedback from the receiver to the transmitters, i.e.,
$P(x_{1,i},x_{2,i}|x_{1}^{i-1},x_{2}^{i-1},y^{i-1})=P(x_{1,i},x_{2,i}|x_{1}^{i-1},x_{2}^{i-1})$,
Equation (\ref{e_memoryless_def}) implies that
\begin{equation}\label{e_memoryless_def2}
P(y_i|x_{1}^{n},x_{2}^{n},y^{i-1})=P(y_i|x_{1,i},x_{2,i}).
\end{equation}

\begin{definition}\label{def_crib}
A $(2^{nR_1},2^{nR_2},n)$ {\it code} with partial cribbing
, as shown in Fig. \ref{f_mac_crib2},
consists at time $i$ of an encoding function at Encoder 1
\begin{eqnarray}
&\text{Case A, B,}&f_{1,i}:\{1,...,2^{nR_1}\}\times  \mathcal Z_2^{i-1} \mapsto \mathcal \mathcal
X_{1,i},
\end{eqnarray}
and an encoding function at Encoder 2 that changes according to the following case settings
\begin{eqnarray}
&\text{Case A}& f_{2,i}:\{1,...,2^{nR_2}\}\times  \mathcal Z_1^{i-1} \mapsto \mathcal
\mathcal X_{1,i},\nonumber \\
&\text{Case B}& f_{2,i}:\{1,...,2^{nR_2}\}\times  \mathcal Z_1^{i} \mapsto \mathcal \mathcal
X_{1,i},
\end{eqnarray}
and a decoding function,
\begin{equation}
g:\mathcal Y^n \mapsto \{1,...,2^{nR_1}\} \times
\{1,...,2^{nR_2}\}.
\end{equation}
\end{definition}

The {\it average probability of error} for
$(2^{nR_1},2^{nR_2},n)$ code is defined as
\begin{equation}
P_e^{(n)}=\frac{1}{2^{n(R_1+R_2)}} \sum_{m_1,m_2} \Pr\{g(Y^n)\neq(m_1,m_2)|(m_1,m_2) \text{
sent}\}.
\end{equation}
A rate $(R_1,R_2)$ is said to be {\it achievable} for the encoder with partial cribbing if there
exists a sequence of $(2^{nR_1},2^{nR_2},n)$ codes with $P_e^{(n)}\to 0$. The {\it capacity
region} of the MAC is the closure of all achievable rates. The following theorem describes the
capacity region of MAC with partial cribbing for two different cases of causality.

Let us define the following regions $\mathcal R_{A},\mathcal R_{B}$, which are contained in
$\mathbb{R}^2_{+}$, namely, contained in  the set of nonnegative two dimensional real numbers.
\begin{equation}\label{e_region2_case_a}
\mathcal R_A=
\left\{\begin{array}{l}R_1\leq H(Z_1|U)+I(X_1;Y|X_2,Z_1,U),\\
R_2\leq H(Z_2|U)+I(X_2;Y|X_1,Z_2,U),\\
R_1+R_2\leq I(X_1,X_2;Y|U,Z_1,Z_2)+H(Z_1,Z_2|U),\\
R_1+R_2\leq I(X_1,X_2;Y), \text{ for }\\
 P(u)P(x_1,z_1|u)P(x_2,z_2|u)P(y|x_1,x_2). \end{array} \right\}
\end{equation}
The region $\mathcal R_B$ is defined with the same set of inequalities as in
(\ref{e_region2_case_a}), but the joint distribution is of the form    \begin{equation}
\label{e_region2_case_b} P(u)P(x_1,z_1|u)P(x_2,z_2|z_1,u)P(y|x_1,x_2).
\end{equation}


\begin{theorem}[Capacity region]\label{t_mac_one_crib}
The capacity regions of the MAC with strictly-causal (Case A), mixed causal
and strictly-causal (Case B) partial cribbing as described in Def. \ref{def_crib} are $\mathcal R_A$, $\mathcal R_B$,  respectively. 
\end{theorem}

\begin{lemma}\label{lemma_cardinality}
To exhaust ${\mathcal R_A}$ and ${\mathcal R_B}$  it is enough to restrict the alphabet of $U$,
 to satisfy
\begin{eqnarray}
|{\mathcal U}|&\leq&\min(|{\mathcal Y}|+3, |{\mathcal X_1}||{\mathcal X_2}|+2)).
\end{eqnarray}
\end{lemma}
The proof of Theorem \ref{t_mac_one_crib} and Lemma \ref{lemma_cardinality} is given in the next
section.
\section{Converse \label{s_converse}}
Here we provide the  converse proof of Theorem  \ref{t_mac_one_crib} for the two cases, A and B.

{\it Converse proof of Case A:}
 Assume that we have a
$(2^{nR_1},2^{nR_2},n)$ code as in  Definition
\ref{def_crib}, Case A. We will show the existence of a joint
distribution $ P(u)P(z_1|u)P(z_2|u)P(x_1|z_1,u)P(x_2|z_2,u)P(y|x_1,x_2)$ that
satisfies the inequalities of (\ref{e_region2_case_a}) within some $\epsilon_n$,
where  $\epsilon_n$ goes to zero as $n\to\infty$. Consider
\begin{eqnarray}\label{e_m1m2_con}
n(R_1+R_2)&=& H(M_{1},M_2)\nonumber \\
&=& H(M_{1},M_2)+H(M_1,M_2|Y^n)-H(M_1,M_2|Y^n)\nonumber \\
&\stackrel{(a)}{=}& I(M_1,M_2;Y^n)+n\epsilon_n \nonumber \\
&\stackrel{(b)}{=}& I(X_1^n,X_2^n;Y^n)+n\epsilon_n \nonumber \\
&\stackrel{}{=}& \sum_{i=1}^n I(X_1^n,X_2^n;Y_i|Y^{i-1})+n\epsilon_n \nonumber \\
&\stackrel{(c)}{\leq}& \sum_{i=1}^n I(X_{1,i},X_{2,i};Y_i)+n\epsilon_n
\end{eqnarray}
where (a) follows from Fano's inequality, (b) from the fact that $(X_1^{n},X_2^n)$ is a deterministic function of $(M_1,M_2)$ and the Markov chain $Y^n-(X_1^n,X_2^n)-(M_1,M_2)$ and (c) from the Markov chain $Y_i-(X_{1,i},X_{2,i})-(X_1^n,X_2^n,Y^{i-1})$.
Now consider,
\begin{eqnarray}\label{e_Hm1m2_con2}
n(R_1+R_2)&=& H(M_{1},M_2)\nonumber \\
&\stackrel{(a)}{=}& H(M_{1},M_2,Z_1^n,Z_2^n)\nonumber \\
&\stackrel{}{=}& H(Z_1^n,Z_2^n)+H(M_{1},M_2|Z_1^n,Z_2^n) \nonumber \\
&\stackrel{(b)}{=}& H(Z_1^n,Z_2^n)+I(M_{1},M_2;Y^n|Z_1^n,Z_2^n)+n\epsilon_n \nonumber \\
&\stackrel{}{=}& H(Z_1^n,Z_2^n)+I(X_{1}^n,X_2^n;Y^n|Z_1^n,Z_2^n)+n\epsilon_n \nonumber \\
&\stackrel{}{=}& \sum_{i=1}^n H(Z_{1,i},Z_{2,i}|Z_1^{i-1},Z_2^{i-1})+I(X_1^n,X_2^n;Y_i|Y^{i-1},Z_1^n,Z_2^n)+n\epsilon_n \nonumber \\
&\stackrel{}{\leq}& \sum_{i=1}^n  H(Z_{1,i},Z_{2,i}|Z_1^{i-1},Z_2^{i-1})+I(X_{1,i},X_{2,n};Y_i|Z_{1}^i,Z_{2}^i)+n\epsilon_n, \nonumber \\
&\stackrel{(c)}{=}& \sum_{i=1}^n H(Z_{1,i},Z_{2,i}|U_i)+I(X_{1,i},X_{2,n};Y_i|Z_{1,i},Z_{2,i},U_i)+n\epsilon_n,
\end{eqnarray}
where (a) follows from the fact that $(Z_1^n,Z_2^n)$ are a deterministic function of $(M_1,M_2)$, (b) from Fano's inequality, and (c) from the following definition of a random variable
\begin{eqnarray}
U_i&\triangleq& (Z_1^{i-1},Z_2^{i-1})\label{e_u_case_a}.
\end{eqnarray}
Furthermore, consider
\begin{eqnarray}\label{e_Hm1_con}
nR_1&\stackrel{}{=}& H(M_{1})\nonumber \\
&\stackrel{(a)}{=}& H(M_{1}|M_2)\nonumber \\
&\stackrel{(b)}{=}& H(M_{1},Z_1^n|M_2)\nonumber \\
&\stackrel{}{=}& H(Z_1^n|M_2)+H(M_{1}|Z_1^n,M_2)\nonumber \\
&\stackrel{}{=}& H(Z_1^n|M_2)+ H(M_1|M_2,Z_1^n)+H(M_1|Y^n,M_2,Z_1^n)-H(M_1|Y^n,M_2,Z_1^n)\nonumber \\
&\stackrel{(c)}{=}& \sum_{i=1}^n H(Z_{1,i}|Z_1^{i-1},M_2)+ I(Y_i;M_1|Y^{i-1},M_2,Z_1^{n})+n\epsilon_n \nonumber \\
&\stackrel{(d)}{=}& \sum_{i=1}^n H(Z_{1,i}|Z_1^{i-1}, Z_2^{i-1},M_2)+I(Y_i;M_1,X_{1,i}|Y^{i-1},M_2,X_{2,i},Z_1^{n},Z_2^{n})+n\epsilon_n \nonumber \\
&\stackrel{(e)}{\leq}& \sum_{i=1}^n H(Z_{1,i}|Z_1^{i-1}, Z_2^{i-1})+I(Y_i;X_{1,i}|X_{2,i},Z_1^{i},Z_2^{i})+n\epsilon_n \nonumber \\
&\stackrel{}{=}& \sum_{i=1}^n H(Z_{1,i}|U_i)+I(Y_i;X_{1,i}|X_{2,i},U_i,Z_{1,i})+n\epsilon_n
\end{eqnarray}
where (a) follows from the fact that the messages $M_1$ and $M_2$ are independent of each other,
(b) follows from the fact that $Z_1^n$ is a deterministic function of $(M_1,M_2)$,  (c) follows
from Fano's inequality, and (d) from the fact that $X_{1,i}$ is a deterministic functions of
$(M_1,Z_2^{i-1})$ and $X_{2,i}$ is a deterministic function of $(M_2,Z_1^{i-1})$. Step (e)
follows from the Markov chain $Y_i-(X_{1,i},X_{2,i},Z^{n})-(M_1,M_2,Y^{i-1})$ and the fact that
conditioning reduces entropy. Similarly to (\ref{e_Hm1_con}) we obtain
\begin{eqnarray}\label{e_Hm2_con}
nR_2&\stackrel{}{\leq}& \sum_{i=1}^n H(Z_{2,i}|U_i)+I(Y_i;X_{2,i}|X_{1,i},U_i,Z_{2,i})+n\epsilon_n.
\end{eqnarray}
Now let us verify that the three  Markov chains $Z_{1,i}-U_i-Z_{2,i}$,
$X_{1,i}-(U_i,Z_{1,i})-(X_{2,i})$, and  $X_{2,i}-(U_i,Z_{2,i})-(X_{1,i})$ hold. The  first Markov
chain is due to the Markov $(M_1,Z_2^{i-1})-(Z_1^{i-1},Z_2^{i-1})-(M_2,Z_1^{i-1})$ or
equivalently $M_1-(Z_1^{i-1},Z_2^{i-1})-M_2$ and the second Markov chain is due to the Markov
chain $(M_1,Z_2^{i-1})-(Z_1^{i},Z_2^{i-1})-(M_2,Z_1^{i-1})$ or equivalently
$M_1-(Z_1^{i},Z_2^{i-1})-M_2$. The Markov chain follows from the joint distribution
$P(m_1,m_2,z_1^n,z_2^n)=P(m_1)P(m_2)\prod_{i=1}^nP(z_{1,i}|z_{2}^{i-1},m_1)\prod_{i=1}^nP(z_{2,i}|z_1^{i-1},m_2)$
and the observation that
\begin{eqnarray}
P(m_1|z_1^n,z_2^n,m_2)&=&\frac{P(m_1)P(m_2)\prod_{i=1}^nP(z_{1,i}|z_{2}^{i-1},m_1)\prod_{i=1}^nP(z_{2,i}|z_1^{i-1},m_2)}{\left(P(m_2)\prod_{i=1}^nP(z_{2,i}|z_1^{i-1},m_2)\right)
\sum_{m_1}P(m_1)\prod_{i=1}^nP(z_{1,i}|z_{2}^{i-1},m_1)}\nonumber\\
&=&\frac{\prod_{i=1}^nP(z_{1,i}|z_{2}^{i-1},m_1)}{
\sum_{m_1}P(m_1)\prod_{i=1}^nP(z_{1,i}|z_{2}^{i-1},m_1),}
\end{eqnarray}
does not depend on $m_2$.
 The third Markov chain is an exchange between the indexes $1$ and
$2$, namely, $M_1,X_{1,i},Z_{1,i}$ is exchanged with $M_2,X_{2,i},Z_{2,i}$, respectively.
Finally, let $Q$ be a random variable independent of $(X_1^n,X_2^n,Y^n)$, and uniformly
distributed over the set $\{1,2,3,..,n\}$. We define the random variables $U\triangleq(Q,U_Q)$
and obtain that the region given in (\ref{e_region2_case_a}) is an outer bound to any achievable
rate.
 \hfill\QED

Once Case A is proved,  Case B follows straightforwardly using the following modification.

{\it Converse proof for Case B:} We repeat the same converse as for Case A, except that in the
final step we need to show the Markov chain $X_{2,i}-(U_i,Z_{1,i},Z_{2,i})-X_{1,i}$ rather than
$X_{2,i}-(U_i,Z_{2,i})-X_{1,i}$ as in Case A. Since for case B the Markov chain
$(M_2,Z_1^i)-(Z_1^i,Z_2^i)-M_1$ holds it follows that
$X_{2,i}-(M_2,Z_1^i)-(U_i,Z_i)-(M_1,Z_2^{i-1})-X_{1,i}$ holds too.\hfill \QED


Now we prove Lemma  \ref{lemma_cardinality} which allows us to bound the cardinality of the
auxiliary random variable $U$ without decreasing the rate regions $\mathcal R_{A},\mathcal
R_{B}$.

{\it Proof of Lemma \ref{lemma_cardinality}:}
 We invoke the support lemma~\cite[pp. 310]{Csiszar81}.
The external random variable $U$ must have $|{\mathcal Y}|-1$ letters to preserve $P(y)$ plus
four more to preserve the expressions $H(Z_1|U)+I(X_1;Y|X_2,Z_1,U)$,
$H(Z_2|U)+I(X_2;Y|X_1,Z_2,U)$, $I(X_1,X_2;Y|U,Z_1,Z_2)+H(Z_1,Z_2|U)$, and $H(Y|X_1,X_2,U,)$.
Alternatively, the external random variable $U$ must have $|{\mathcal X_1}||{\mathcal X_2}|-1$
letters to preserve $P(x_1,x_2)$ and three more to preserve the expressions
$H(Z_1|U)+I(X_1;Y|X_2,Z_1,U)$, $H(Z_2|U)+I(X_2;Y|X_1,Z_2,U)$,
$I(X_1,X_2;Y|U,Z_1,Z_2)+H(Z_1,Z_2|U)$. Hence the cardinality of $U$ may be bounded by
$\min(|{\mathcal Y}|+3, |{\mathcal X_1}||{\mathcal X_2}|+2)).$ \hfill \QED

\section{Achievability proof of Theorem \ref{t_mac_one_crib}\label{s_achievability}}
In this section we provide the  achievability proof of Theorem  \ref{t_mac_one_crib} for the two
cases, A and B. Throughout the achievability proofs in the paper we use the definition of a
strong typical set. The set $T^{(n)}_\epsilon(X,Y,Z)$ of $\epsilon$-typical $n-$sequences is
defined by $\{(x^n,y^n,z^n):\frac{1}{n}N(x,y,z|x^n,y^n,z^n)-p(x,y,z)|\leq \epsilon p(x,y,z)
\forall (x,y,z)\in \mathcal X\times \mathcal Y\times \mathcal Z\}$, where $N(x,y,z|x^n,y^n,z^n)$
is the number of appearances of $(x,y,z)$ in the $n-$sequnce $(x^n,y^n,z^n)$. Additionally, we
will use the following well-known lemma
\cite{Csiszar81,CovThom06,Kramer07BookMultiUser,ElGammalKim10LectureNotes},
\begin{lemma}[Joint typicality lemma]\label{l_typical}
Consider a joint distribution $P_{X,Y,Z}$ and suppose $(x^n,y^n)\in
T_\epsilon^{(n)}(X,Y)$. Let $\tilde Z^n$ be distributed according to
$\prod_{i=1}^n P_{Z|X}(\tilde z_i|x_i)$. Then,
\begin{equation}
\Pr\{(x^n,y^n,\tilde Z^n)\in T^{(n)}_{\epsilon}(X,Y,Z)\}\leq
2^{-n(I(Y;Z|X)-\delta(\epsilon))},
 \end{equation}
 where $\lim_{\epsilon\to 0}\delta(\epsilon)= 0$.
\end{lemma}

For the achievability proof, we use the rate-splitting coding technique in addition to the
techniques used by Willems \cite{Willems85_cribbing_encoders}, i.e.,  block Markov coding,
super-position coding, Shannon's strategies and backward decoding. The rate splitting technique
introduces additional rate variables which are redundant and we eliminate them using the
Fourier$-$Motzkin elimination.

{\it Achievability Proof of Case A:} Let us split rate $R_1$ into two rates $R_1'$ and $R_1''$
such that $R_1=R_1'+R_1''$ and similarly $R_2$ into $R_2'$ and $R_2''$ such that
$R_2=R_2'+R_2''$. Let $m_1'\in [1,...,2^{nR_1'}]$, $m_1''\in [1,...,2^{nR_1''}]$, $m_2'\in
[1,...,2^{nR_2'}]$, and $m_2''\in [1,...,2^{nR_2''}]$. Note that there is a one-to-one mapping
between $(m_1',m_1'')$ and  $m_1$ and between $(m_2',m_2'')$  and  $m_2$.

   {\it Code construction:}
 Divide a block of length $Bn$ into $B$ blocks of length $n$. We use random coding to generate
independently the code for each subblock $b$.  Construct $2^{n(R_1'+R_2')}$ codewords $u^n$
according to i.i.d. $\sim P(u)$. For every codeword $u^n$ construct  $2^{nR_1'}$ codewords
$z_1^n$ according to i.i.d. $\sim P(z_1|u)$ and similarly $2^{nR_2'}$ codewords $z_2^n$ according
to i.i.d. $\sim P(z_2|u)$. Furthermore, generate $2^{nR_1''}$ codewords $x_1^n$ according to
i.i.d. $\sim P(x_1|z_1,u)$ and similarly $2^{nR_2''}$ codewords $x_2^n$ according to i.i.d. $\sim
P(x_2|z_2,u)$ . The Markov structure of the code is
  \begin{eqnarray}
  x_1^n &\text{is determined by}& (m'_{1,b},m''_{1,b}) \text{ conditioned on  } (m'_{1,b-1},m'_{2,b-1})\nonumber \\
  x_2^n &\text{is determined by}& (m'_{2,b},m''_{2,b}) \text{ conditioned on  } (m'_{1,b-1}, m'_{2,b-1}).
  \end{eqnarray}

{\it Encoder:} At block $b\in[1,...,B]$ encode the message $(m'_{1,b-1},m'_{2,b-1})\in
[1,..,2^{n(R'_1+R'_2)}]$ using  $u^n(m'_{1,b-1},m'_{2,b-1})$, encode $m'_{1,b}$ conditioned on
$(m'_{1,b-1},m'_{2,b-1})$ using  $z_1^n(u^n,m'_{1,b})$, and encode $m''_{1,b}$ conditioned on
$(m'_{1,b},m'_{1,b-1},m'_{2,b-1})$ using  $x_1^n(z_1^n,u^n,m''_{1,b})$. Similarly, encode
$m'_{2,b}$ conditioned on $(m'_{1,b-1},m'_{2,b-1})$ using  $z_2^n(u^n,m'_{1,b})$, encode
$m''_{2,b}$ conditioned on $(m'_{2,b},m'_{1,b-1},m'_{2,b-1})$ using
$x_2^n(z_2^n,u^n,m''_{2,b})$. We assume that $m'_{1,0}=m'_{2,0}=1$ and $m'_{1,b}=m'_{2,b}=1$
which allow a backward decoding as explained next.

{\it Decoder:} The receiver waits till the end of the block $Bn$ and starts decoding each message in the sub-blocks going backwards $b\in[B,B-1, B-2,...,1]$. At block $b$, we assume that $(m'_{1,b},m'_{2,b})$ is already known to the receiver from block $b+1$ and it needs to decode  , $m'_{1,b-1}$, $m'_{2,b-1}$, $m''_{2,b}$ and $m''_{2,b}$. The decoder uses joint typicality decoding, hence at block $b$ it looks for   $(\hat m'_{1,b-1},\hat m'_{2,b-1})$, $\hat m''_{2,B}$ and $\hat m''_{2,B}$ for which
\begin{equation}
(u^n(\hat m'_{1,b-1},\hat m'_{2,b-1}), z_1^n(u^n, m'_{1,b}), z_2^n(u^n, m'_{2,b}), x_1^n(z_1^n,u^n,\hat m''_{1,b}),x_2^n(z_2^n,u^n,\hat m''_{2,b})\in T_{\epsilon}^{(n)}(U,Z_1,Z_2,X_1,X_2,Y).
\end{equation}
If no such triplet, or more than one such triplet  is found, an error is declared at block $b$
and therefore at the whole superblock $nB$ (we consider $(\hat m'_{1,b-1},\hat m'_{2,b-1})$  as
one index in $[1,...,2^{nR_1'+nR_2'}]$. The estimated message at block $b$ sent from Encoder 1 is
$(\hat m_{1,a},\hat m_{1,b})$, and the estimated message transmitted from Encoder 2 is $(\hat
m_{2,a},\hat m_{2,b})$.

{\it Error analysis:}
The following lemma will enable us to bound the probability of error of the super-block $nB$ by bounding the probability of error of each block.
\begin{lemma}\label{l_conditional_upper}
Let $\{A_j\}_{j=1}^J$ be a set of events and let $A_j^c$ denotes the complement of the event $A_j$. Then
\begin{equation}
P(\bigcup_{j=1}^J A_j)\leq \sum_{j=1}^nP(A_j|\bigcap_{i=1}^{j-1}A_i^c)=\sum_{j=1}^nP(A_j|A_1^c,A_2^c,...,A_{j-1}^c).
\end{equation}
\end{lemma}
\begin{proof} For simplicity let us assume that $J=3$. In a straightforward manner the proof extends to any number of sets $J$.
 For any three sets of events $A_1,A_2,A_3$ we have
 \begin{eqnarray}
 P(A_1 \cup A_2 \cup A_3)&=&P\left(A_1 \cup (A_2\cap A_1^c)  \cup (A_3 \cap A_1^c \cap A_2^c)\right)\nonumber \\
  &=&P(A_1) + P(A_2\cap A_1^c) + P(A_3 \cap A_1^c \cap A_2^c)\nonumber \\
  &\leq&P(A_1) + \frac{P(A_2\cap A_1^c)}{P(A_1^c)} + \frac{P(A_3 \cap A_1^c \cap A_2^c)}{P(A_1^c \cap A_2^c)} \nonumber \\
   &\stackrel{}{=}& P(A_1) + P(A_2| A_1^c) + P(A_3| A_1^c \cap A_2^c)\nonumber \\
   &\stackrel{}{=}& P(A_1) + P(A_2| A_1^c) + P(A_3| A_1^c,A_2^c)
 \end{eqnarray}
\end{proof}

 Using Lemma \ref{l_conditional_upper} we bound the probability of error in the supper block $Bn$ by the sum of the probability of having an error in each block $b$ given that in previous blocks $(b+1,...,B)$ the messages were decoded correctly.

First let us bound the probability that for some $b$, Transmitter 1 decodes the message $m'_{2,b}$  incorrectly or Transmitter 2 decodes the message $m'_{1,b}$  incorrectly  at the end of block $b$. Using Lemma \ref{l_conditional_upper} it suffices to show that the probability of error-decoding  in each block $b$ goes to zero, assuming that all previous messages in block $(1,2,...,b-1)$ were decoded correctly.

Let $E_{1,b}$ be the event that Transmitter 1 has an error in decoding  $m'_{2,b}$ and let
$E_{2,b}$ be the event that Transmitter 2 has an error in decoding  $m'_{1,b}$. The term
$P(E_{1,b}\cup E_{2,b}|E_{0,b-1}^c)$ is the probability  that Transmitter 1 or 2 incorrectly
decoded  $m'_{2,b}$ and $m'_{1,b}$, respectively, given that $m'_{1,b-1}$ and $m'_{2,b-1}$ were
decoded correctly. Without loss of generality let's assume that $m'_{1,b}=m'_{2,b}=1.$ An error
occurs if and only if there is another message $m'_{1,b}>1$ that maps to the same codeword as
$z_1^n(1,u^n)$ or there is another message $m'_{2,b}>1$ that maps to the same codeword as
$z_2^n(1,u^n)$. The probability that $z_1^n(i,u^n)=z_1^n(1,u^n)$ where $i>1$ and where
$(z_1^n(1,u^n),u^n)\in T_{\epsilon}^{(n)}(Z_1,U)$ is bounded by
$2^{-n(H(Z_1|U)-\delta(\epsilon))}$, where $\delta(\epsilon)$ goes to zero as $\epsilon$ goes to
zero. Hence
\begin{eqnarray}
P(E_{1,b}\cup E_{2,b}|E_{1,b-1}^c, E_{2,b-1}^c)&\stackrel{(a)}{\leq}& P(E_{1,b}|E_{1,b-1}^c, E_{2,b-1}^c)+ P(E_{2,b}|E_{1,b-1}^c, E_{2,b-1}^c) \nonumber \\
&\leq& \sum_{i=2}^{2^{nR_1'}}2^{-n(H(Z_1|U)-\delta(\epsilon))}+\sum_{i=2}^{2^{nR_2'}}2^{-n(H(Z_2|U)-\delta(\epsilon))}\nonumber\\
&\leq& 2^{n(R_1'-n(H(Z_1|U))+\delta(\epsilon))}+2^{n(R_2'-n(H(Z_2|U)+\delta(\epsilon)))},
\end{eqnarray}
where inequality (a) follows from the union bounds. Now we bound the probability that the
receiver decodes the messages $(m'_{1,b-1}, m'_{2,b-1})$, or $ m''_{2,b}$ or $ m''_{2,b}$
incorrectly at block $b$ given that at block $b+1$ the messages $(m'_{1,b}, m'_{2,b})$  were
decoded correctly and given that Transmitter 1 and 2 encodes the right messages $(m'_{1,b-1},
m'_{2,b-1})$ in block $b$. Without loss of generality assume $(m'_{1,b-1}, m'_{2,b-1})=1$ (for
simplicity we index both messages by one index), $ m''_{2,b}=1$ and $ m''_{2,b}=1$ . Let us
define the event
\begin{equation}
E_{i,j,k,b}\triangleq\left\{
\left(u^n(i),z_1^n(u^n,m'_{1,b}),z_2^n(u^n,m'_{2,b}),x_1^n(u^n,z_1^n,j),x_2^n(u^n,z_2^n,k),y^n\right)\in
T_{\epsilon}^{(n)}(U,Z_1,Z_2,X_1,X_2,Y)\right\}.
\end{equation}
An error occurs if either the correct codewords are not jointly
typical with the received sequences, i.e., $E_{1,1,1,b}^c$, or there
exists a different $(i, j,k)\neq(1,1,1)$ such that $E_{i,j,k,b}$
occurs. Let $P_{e,b}^{(n)}$ be the error-decoding at block $b$ given that in blocks $(b+1,...,B)$ there was no error-decoding. From the union of bounds we obtain that
\begin{equation}\label{e_pe}
P_{e,b}^{(n)}\leq \Pr(E_{1,1,1,b}^c)+\sum_{i=1,j=1,k>1}
\Pr(E_{i,j,k,b})+\sum_{i=1,j>1,k=1} \Pr(E_{i,j,k,b})+\sum_{i=1,j>1,k>1}
\Pr(E_{i,j,k,b})+\sum_{i>1,j\geq1,k\geq 1} \Pr(E_{i,j,k,b}).
\end{equation}
Now let us show that each term in (\ref{e_pe}) goes to zero as the
blocklength of the code $n$ goes to infinity.
\begin{itemize}
\item Upper-bounding $\Pr(E_{1,1,1}^c)$:
Since we assume that the Transmitter 1 and 2 encode the right $(m'_{1,b-1}, m'_{2,b-1})$  and the receiver decoded the right $(m'_{1,b}, m'_{2,b})$ in block $b+1$, by the LLN  $\Pr(E_{1,1,1,b}^c)\to 0.$
\item Upper-bounding $\sum_{i=1,j=1,k>1}
\Pr(E_{i,j,k})$:  The probability that $Y^n$, which is generated according to
$P(y|x_1,x_2)=P(y|x_1,x_2,u,z)$, is jointly typical with $x_2^n$, which was generated according
to $P(x_2|z_2,u)=P(x_2|u,z_1,z_2,x_1)$, where $(x_1^n,z_1^n,z_2^n,u^n)\in
T^{(n)}_\epsilon(X_1,Z_1,Z_2,U)$ is bounded by (Lemma \ref{l_typical})
\begin{equation}
\Pr\{(x_1^n,z_1^n,X_2^n,z_2^n,u^n,Y^n)\in
T^{(n)}_\epsilon|(x_1^n,z_1^n,z_2^n,u^n)\in T^{(n)}_\epsilon\}\leq
2^{-n(I(X_2;Y|X_1,Z_2,U)-\delta(\epsilon))}.
\end{equation}
Hence, we obtain
\begin{eqnarray}\label{e_k}
\sum_{i=1,j=1,k>1}
\Pr(E_{i,j,k,b})&\leq&2^{nR''_2}2^{-n(I(X_2;Y|X_1,Z_1,Z_2,U)-\delta(\epsilon))}.
\end{eqnarray}

\item Upper-bounding $\sum_{i=1,j>1,k=1} \Pr(E_{i,j,k,b})$:
Similarly, to (\ref{e_k}) we obtain
\begin{eqnarray}
\sum_{i=1,j>1,k=1}
\Pr(E_{i,j,k,b})&\leq&2^{nR''_1}
2^{-n(I(X_1;Y|X_2,Z_1,Z_2,U)-\delta(\epsilon))}.
\end{eqnarray}

\item  Upper-bounding $\sum_{i=1,j>1,k>1}
\Pr(E_{i,j,k,b})$ by
\begin{eqnarray}
\sum_{i=1,j>1,k>1}
\Pr(E_{i,j,k,b})&\leq&2^{n(R_2''+R''_1)}
2^{-n(I(X_2,X_1;Y|Z_1,Z_2,U)-\delta(\epsilon))}.
\end{eqnarray}

\item  Upper-bounding $\sum_{i>1,j\geq1,k\geq 1} \Pr(E_{i,j,k,b})$ by
\begin{eqnarray}\label{e_u5}
\sum_{i>1,j\geq1,k\geq 1}
\Pr(E_{i,j,k,b})&\leq&2^{n(R''_1+R'_1+R_2)}
2^{-n(I(X_2,X_1,U,Z_1,Z_2;Y)-\delta(\epsilon))}\nonumber\\
&=&2^{n(R_1+R_2-I(X_2,X_1;Y)-\delta(\epsilon))}
\end{eqnarray}

\end{itemize}

To summarize we obtained that if $R_1'=R_1-R_1''$, $R_1''$,  $R_2'=R_2-R_2''$, $R_2''$ and $R_2$ satisfy
\begin{eqnarray}\label{e_achie_in}
R_1-R_1''&\leq& H(Z_1|U),\nonumber \\
R_2-R_2''&\leq& H(Z_2|U),\nonumber \\
R''_1&\leq& I(X_1;Y|X_2,Z_1,U),\nonumber\\
R''_2&\leq& I(X_2;Y|X_1,Z_2,U),\nonumber\\
R''_1+R''_2&\leq& I(X_1,X_2;Y|Z_1,Z_2,U),\nonumber\\
R_1+R_2&\leq& I(X_2,X_1;Y),
\end{eqnarray}
then there exists a sequence of code with a probability of error that goes to zero as the block
length goes to infinity. Using Fourier$-$Motzkin elimination \cite{LecturesConvexSets10} first
for $R_1''$ we obtain
\begin{eqnarray}
R_1- H(Z_1|U)&\leq& I(X_1;Y|X_2,Z_1,U) ,\nonumber \\
R_2-R_2''&\leq& H(Z_2|U),\nonumber \\
R''_2&\leq& I(X_2;Y|X_1,Z_2,U),\nonumber\\
R_1- H(Z_1|U)+R''_2&\leq& I(X_1,X_2;Y|Z_1,Z_2,U),\nonumber\\
R_1+R_2&\leq& I(X_2,X_1;Y),
\end{eqnarray}
and applying Fourier$-$Motzkin elimination also  for $R_2''$ we obtain
\begin{eqnarray}
R_1- H(Z_1|U)&\leq& I(X_1;Y|X_2,Z_1,U) ,\nonumber \\
R_2-H(Z_2|U)&\leq&I(X_2;Y|X_1,Z_2,U) ,\nonumber \\
R_1- H(Z_1|U)+R_2-H(Z_2|U)&\leq& I(X_1,X_2;Y|Z_1,Z_2,U),\nonumber\\
R_1+R_2&\leq& I(X_2,X_1;Y),
\end{eqnarray}
which is equivalent to the region of Case A in (\ref{e_region2_case_a}).\hfill \QED

 {\it Achievability for Case B:} The achievability of case B is very similar to case A, only that the codewords of $X_2$ needs to be generated according to shannon strategy  (or a code-trees ) rather than codewords. This is due to the fact that $Z_{1,i}$ is known causally and $X_2$ is generated according to a distribution $P(x_2|u,z_1)$. \hfill \QED


\section{common message \label{s_commom_message}}
Let us now consider the case where a common message, $m_0\in\{1,2,...,2^{nR_0}\}$, is known to
encoders 1 and 2 and needs to be transmitted to the decoder in addition to the private messages
$m_1,m_2$.  Hence Encoder 1 is given by the function
\begin{eqnarray}
&\text{Case A, B,}&f_{1,i}:\{1,...,2^{nR_0}\}\times\{1,...,2^{nR_1}\}\times  \mathcal Z_2^{i-1}
\mapsto \mathcal
\mathcal X_{1,i},
\end{eqnarray}
and  Encoder 2 is given by the functions
\begin{eqnarray}
&\text{Case A,}& f_{2,i}:\{1,...,2^{nR_0}\}\times\{1,...,2^{nR_2}\}\times  \mathcal Z_1^{i-1}
\mapsto \mathcal
\mathcal X_{1,i},\nonumber \\
&\text{Case B,}& f_{2,i}:\{1,...,2^{nR_0}\}\times\{1,...,2^{nR_2}\}\times  \mathcal Z_1^{i}
\mapsto \mathcal \mathcal X_{1,i}.
\end{eqnarray}
Remarkably, no additional auxiliary random variable is needed to characterizes the capacity
region, since the partial cribbing is used for generating a common message. Let the rate regions
$\mathcal R_A^0$ and $\mathcal R_B^0$ be defined exactly as $\mathcal R_A$ and $\mathcal R_B$
only that the last inequality in (\ref{e_region2_case_a}), i.e., $R_1+R_2\leq I(X_1,X_2;Y)$,  is
replaced by
\begin{equation}
R_0+R_1+R_2\leq I(X_1,X_2;Y).
\end{equation}

\begin{theorem}[Capacity region in the case of a common message]\label{t_mac_crib_common}
The capacity regions of the MAC with strictly-causal (Case A), and mixed causal and
strictly-causal (Case B) partial cribbing with a common message  are
$\mathcal R_A^0$ and $\mathcal R_B^0$, respectively. 
\end{theorem}

Note that if there is no cribbing, i.e., $Z_1$ and $Z_2$ are constant, we obtain the capacity
region of the MAC with a common message as derived by Slepian and Wolf \cite{Slepian_Wolf_MAC73}.
We sketch here only the differences between the proof of Theorem \ref{t_mac_crib_common} and
Theorem \ref{t_mac_one_crib}.

{\it Proof of Theorem \ref{t_mac_crib_common}}

{\bf Converse:} Similar to the sequence of inequalities in (\ref{e_m1m2_con}) we have
\begin{eqnarray}\label{e_m1m2_com}
n(R_0+R_1+R_2)&=& H(M_0,M_{1},M_2)\nonumber \\
&\stackrel{}{\leq}& \sum_{i=1}^n I(X_{1,i},X_{2,i};Y_i)+n\epsilon_n.
\end{eqnarray}
Adding conditioning on $M_0$ in the sequence of inequalities (\ref{e_Hm1m2_con2}) we obtain
\begin{eqnarray}\label{e_Hm1m2_com}
n(R_1+R_2)&=& H(M_{1},M_2)\nonumber \\
&=& H(M_{1},M_2|M_0)\nonumber \\
&\stackrel{}{\leq}& \sum_{i=1}^n  H(Z_{1,i},Z_{2,i}|Z_1^{i-1},Z_2^{i-1},M_0)+I(X_{1,i},X_{2,n};Y_i|Z_{1}^i,Z_{2}^i,M_0)+n\epsilon_n, \nonumber \\
&\stackrel{}{=}& \sum_{i=1}^n
H(Z_{1,i},Z_{2,i}|U_i)+I(X_{1,i},X_{2,n};Y_i|Z_{1,i},Z_{2,i},U_i)+n\epsilon_n,
\end{eqnarray}
where the last step is due to the new definition of $U_i$ as
\begin{eqnarray}
U_i&\triangleq& (M_0,Z_1^{i-1},Z_2^{i-1})\label{e_u_case_com}.
\end{eqnarray}
Similarly, adding conditioning on $M_0$ in the sequence of inequalities (\ref{e_Hm1_con}) we
obtain
\begin{eqnarray}\label{e_Hm1_com}
nR_1&\stackrel{}{=}& H(M_{1})\nonumber \\
&\stackrel{}{=}& H(M_{1}|M_2,M_0)\nonumber \\
&\stackrel{}{\leq}& \sum_{i=1}^n H(Z_{1,i}|Z_1^{i-1}, Z_2^{i-1},M_0)+I(Y_i;X_{1,i}|X_{2,i},Z_1^{i},Z_2^{i},M_0)+n\epsilon_n \nonumber \\
&\stackrel{}{=}& \sum_{i=1}^n H(Z_{1,i}|U_i)+I(Y_i;X_{1,i}|X_{2,i},U_i,Z_{1,i})+n\epsilon_n.
\end{eqnarray}
In a similar way, we obtain the inequality for $R_2$ as in (\ref{e_Hm2_con}).

{\bf Achievability:} The achievability proof is  similar to that in Theorem \ref{t_mac_one_crib}
except that we generate $2^{n(R_1'+R_2'+R_0)}$ codewords $u^n$ according to i.i.d. $\sim P(u)$,
rather than $2^{n(R_1'+R_2')}$, and wherever we have in the achievability proof of Theorem
\ref{t_mac_one_crib} $(m'_{1,b-1},m'_{2,b-1})$ we should now have $(m_0,m'_{1,b-1},m'_{2,b-1})$.
Hence we obtain the same sequence of inequalities as in (\ref{e_achie_in}) except that in the
last inequality which corresponds to an error in all messages we have
\begin{equation}R_0+R_1+R_2\leq I(X_2,X_1;Y).
\end{equation}

\hfill \QED
\section{Special case of partial cribbing: Semi-deterministic relay channel\label{s_semi_deterministic}}
As a special case of the partial cribbing encoders, let us consider the case where Encoder 2 has
no message to send, i.e., $R_2=0$, and only Encoder 2 cribs from Encoder 1, i.e., $Z_2$ is a
constant.  We show here that indeed the region obtained via partial cribbing when $R_2=0$ and the
region obtained via semi-deterministic relay channel coincide.

{\it Case A, semi-deterministic relay with  a delay:} This case become a special case of the
semi-deterministic relay channel which was introduced and solved by El-Gamal
\cite{ElGamalrefA82SemiDetereministicRelay}, where Encoder 2 plays the role of the relay. In such
a case the region $\mathcal R_A$ becomes
\begin{equation}
\mathcal R_A=\left\{\begin{array}{l}R_1\leq H(Z_1|U)+I(X_1;Y|X_2,Z_1,U),\\
R_1\leq I(X_1,X_2;Y|U,Z_1)+H(Z_1|U),\\
R_1\leq I(X_1,X_2;Y), \text{ for }\\
 P(u)P(z_1|u)P(x_1|z_1,u)P(x_2|u)P(y|x_1,x_2). \end{array} \right\}
\end{equation}
Clearly, $H(Z_1|U)+I(X_1;Y|X_2,Z_1,U)\leq I(X_1,X_2;Y|U,Z_1)+H(Z_1|U)$ hence the region we
obtained is $R_1\leq \min(H(Z_1|U)+I(X_1;Y|X_2,Z_1,U),I(X_1,X_2;Y))$ for some
$P(u)P(z_1|u)P(x_1|z_1,u)P(x_2|u)$. Now consider
\begin{eqnarray}
H(Z_1|U)+I(X_1;Y|X_2,Z_1,U)&\stackrel{(a)}{=}&H(Z_1|U,X_2)+I(X_1;Y|X_2,Z_1,U)\nonumber \\
&\stackrel{(b)}{\leq}&H(Z_1|X_2)+I(X_1;Y|X_2,Z_1),\label{e_upper_semi}
\end{eqnarray}
where step (a) follows from the Markov chain  $X_2-U-Z_1$ and step (b) from the fact that
conditioning reduces entropy and from the Markov chain $Y-(X_1,Z_1,X_2)-U$. By choosing $U=X_2$
we obtain the upper bound of (\ref{e_upper_semi}) and the expression $I(X_1,X_2;Y)$ does not
decrease. Hence the capacity region is
\begin{equation}\label{e_semi_deterministic_with_delay}
R_1\leq \min(H(Z_1|X_2)+I(X_1;Y|X_2,Z_1),I(X_1,X_2;Y))
\end{equation}
 for some $P(x_1,x_2)$. Eq. (\ref{e_semi_deterministic_with_delay}) coincides with the result in
 \cite{ElGamalrefA82SemiDetereministicRelay}.

{\it Case B, semi-deterministic relay without delay:} In this case $\mathcal R_B$ become the set
of rates $R_1$ that satisfies
\begin{equation}\label{e_semi_without_delay}
R_1\leq \min( H(Z_1|U)+I(X_1;Y|X_2,Z_1,U),I(X_1,X_2;Y))
\end{equation}
for some $P(x_1,z_1,u)P(x_2|u,z_1).$ The case of relays without delay was investigated by
El-Gamal et. al. in \cite{ElGamal07RelayNoDelay} where it was shown that the capacity region for
the semi-deterministic relay without delay which is denoted by $C_{0,\text{semi-det}}$ is
\begin{equation}\label{e_semi_without_delya_capacity_elGammal}
C_{0,\text{semi-det}}=\max_{P(u,x_1), x_2=f(u,z_1)} \min (I(X_1;Y,Z_1|U), I(U,X_1;Y)) .
\end{equation}
At first glance, the expressions in (\ref{e_semi_without_delay}) seem to be different from the
expression in (\ref{e_semi_without_delya_capacity_elGammal}), but with some simple manipulations
one can show that the expression are equivalent. In particular, the first term in
(\ref{e_semi_without_delya_capacity_elGammal}) may be written as
\begin{eqnarray}
I(X_1;Y,Z_1|U)&=&I(X_1;Z_1|U)+I(X_1;Y|U,Z_1)\nonumber\\
&\stackrel{(a)}{=}&H(Z_1|U)+I(X_1;Y|U,Z_1)\nonumber\\
&\stackrel{(b)}{=}&H(Z_1|U)+I(X_1;Y|U,Z_1,X_2),
\end{eqnarray}
where step (a) follows from the fact that $Z_1$ is a function of $X_1$ and step (b) from the fact
that $X_2$ is a function of $(U,Z_1)$. The second term in
(\ref{e_semi_without_delya_capacity_elGammal}) may be written as
\begin{eqnarray}
I(U,X_1;Y)&\stackrel{(a)}{=}&I(U,X_1,X_2;Y)\nonumber\\
&\stackrel{(b)}{=}&I(X_1,X_2;Y),
\end{eqnarray}
where step (a) follows from the fact that $X_2$ is a function of $(U,X_1)$ and step (b) follows
from the  Markov chain $Y-(X_1,X_2)-U$. Now, to conclude that  (\ref{e_semi_without_delay}) and
(\ref{e_semi_without_delya_capacity_elGammal}) are equivalent we need to show that it suffices to
consider only distributions where $X_2$ is a function of $(U,Z_1)$ in
(\ref{e_semi_without_delay}). It follows from \cite[Lemma
1]{chenLeiCuffPermuter_multipleDescription} that there exists a random variable $W$ independent
of $(U,Z_1)$ and satisfies $W-(X_2,U,Z_1)-(Y,X_1)$ such that $X_2$ is a deterministic function of
$(U,Z_1,W)$. Therefore
\begin{eqnarray}
H(Z_1|U)+I(X_1;Y|X_2,Z_1,U)&=&H(Z_1|U,W)+I(X_1;Y|X_2,Z_1,U,W)\nonumber\\
&=&H(Z_1|\tilde U)+I(X_1;Y|X_2,Z_1,\tilde U),
\end{eqnarray}
where $\tilde U=(U,W)$. Hence it suffices to consider $X_2$ that is a function of $(\tilde
U,Z_1)$ and it emerges that (\ref{e_semi_without_delay}) is equivalent to
(\ref{e_semi_without_delya_capacity_elGammal}).

\section{Gaussian MAC with quantized cribbing \label{s_Gaussian_mac}}
In this section we consider the additive Gaussian noise MAC, i.e., $Y=X_1+X_2+W$, where $W$ is a
memoryless Gaussian noise with variance $N$, i.e., $W\sim \text{Norm}(0,N)$. We assume a power
constraints $P_1$ and $P_2$ on the inputs from Encoder 1 and Encoder 2, respectively. If the
encoders do not cooperate than the capacity is given by
\begin{eqnarray}\label{e_gaussian_no_criibing}
R_1&\leq& \frac{1}{2}\log\left(1+\frac{P_1}{N}\right)\nonumber \\
R_2&\leq& \frac{1}{2}\log\left(1+\frac{P_2}{N}\right)\nonumber \\
R_1+R_2&\leq& \frac{1}{2}\log\left(1+\frac{P_1+P_2}{N}\right).
\end{eqnarray}
\begin{figure}[h!]{
\psfrag{B}[][][1]{Encoder 1} \psfrag{D}[][][1]{Encoder 2} \psfrag{m1}[][][1]{$m_1 \ \ \ \ $}
\psfrag{m2}[][][1]{$\in\{1,...,2^{nR_1}\}\ \ \ \ \ \ \ $} \psfrag{m3}[][][1]{$m_2 \ \ \ \ $}
\psfrag{m4}[][][1]{$\in\{1,...,2^{nR_2}\}\ \ \ \ \ \ \ $}  \psfrag{x1}[][][1]{$\ \ \ \ \ \ \ \ \
\ \ \ X_{1,i}(m_1)$} \psfrag{x2}[][][1]{\; \; \; \; \; $X_{2,i}(m_2,Z^{i})$}
 \psfrag{Yi}[][][1]{$Y$} \psfrag{W}[][][1]{Decoder}
\psfrag{t}[][][1]{$$} \psfrag{G2}[][][0.9]{Quantization} \psfrag{G1}[][][0.9]{Scalar}
\psfrag{W1}[][][0.9]{$\ \ \ \ \ \ \ \ \ \ \ \ \ \ \ \ W\sim$Norm$(0,N)$}

\psfrag{Y}[][][1]{$\ \ \ \ \ \ \hat m_1(Y^n)$} \psfrag{Y2}[][][1]{$\ \ \ \ \hat m_2(Y^n)$}
\psfrag{Z}[][][1]{$Z_i$}

\centerline{\includegraphics[width=12cm]{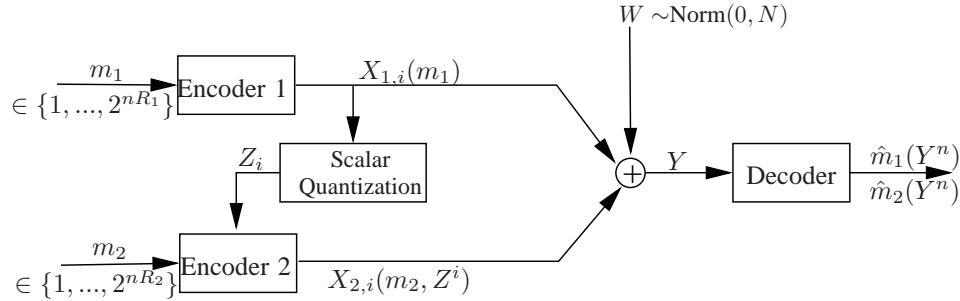}}

\caption{Gaussian MAC with quantized cribbing. The cribbing that Encoder 2 observes is the
quantized signal from Encoder 1. There exist power constraints $\sum_{i=1}^n E[X_{1,i}^2]\leq
P_1$ and $\sum_{i=1}^n E[X_{1,i}^2]\leq P_2$. }\label{f_mac_gauss_channel}
 }\end{figure}
If there is perfect cribbing from Encoder 1 to Encoder 2, either with delay or without the
capacity is the same as if Encoder 2 knows the message of Encoder 1 since Encoder 1 can send the
message in one epoch time. Hence, the capacity is the union over $0\leq \rho\leq 1$ of the
regions
\begin{eqnarray}\label{e_gaussian_perfect_cribbing}
R_2&\leq& \frac{1}{2}\log\left(1+\frac{P_2}{N}(1-\rho^2))\right)\nonumber \\
R_1+R_2&\leq& \frac{1}{2}\log\left(1+\frac{P_1+2\rho \sqrt{P_1P_2}+P_2}{N}\right).
\end{eqnarray}
Now, let us consider the case where Encoder 2 observes  a quantized version of the signal from
Encoder 1 without delay. The setting is depicted in Fig. \ref{f_mac_gauss_channel}. We assume
that the quantizer is a scalar quantizer designed such that under a Gaussain input with variance
$P_1=1$ the discrete values have the same probability (see Fig. \ref{f_quant} for an example of
2-bit quantizer).

\begin{figure}[h!]{
\psfrag{x1}[][][1]{$x_1$} \psfrag{Q1}[][][0.9]{$Q=1$} \psfrag{Q2}[][][0.9]{$Q=2\ \ \ \ $}
\psfrag{Q3}[][][0.9]{$Q=3 \ \ \ \ $} \psfrag{Q4}[][][0.9]{$Q=4$} \psfrag{p(x)}[][][0.9]{$f(x_1)$}
\centerline{\includegraphics[width=8cm]{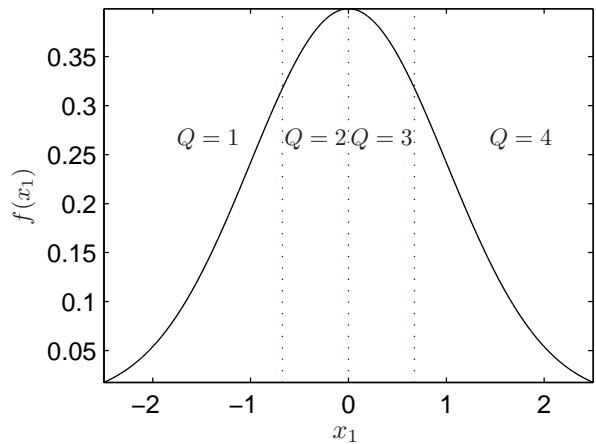}} 

\caption{The 2-bit quantizer's boundaries are designed such that if the input signal has a normal
distribution with variance $P_1=1$ the output values from the quantizer have equal probability.
The input to the 2-bit quantizer is $X_1$ and the output is $Q\in\{1,2,3,4\}.$ \label{f_quant} }
}\end{figure}
 Next, we consider a simple achievable scheme for the Gaussian MAC with a quantizer
cribbing without delay, where the power constraints are $P_1=P_2=1$ and the noise variance is
$N=\frac{1}{2}$. We evaluate the region $\mathcal R_B$ given by (\ref{e_region2_case_a}) and
(\ref{e_region2_case_b}) for the case where $U_1,U_2,Z_2$ are constants, $X_1\sim N(0,1)$, $Z_1$
is a quantized version of $X_1$ such that each value has equal probability.
\begin{figure}[h!]{
\psfrag{R1}[][][1]{$R_1$} \psfrag{R2}[][][1]{$R_2$}
\centerline{\includegraphics[width=9cm]{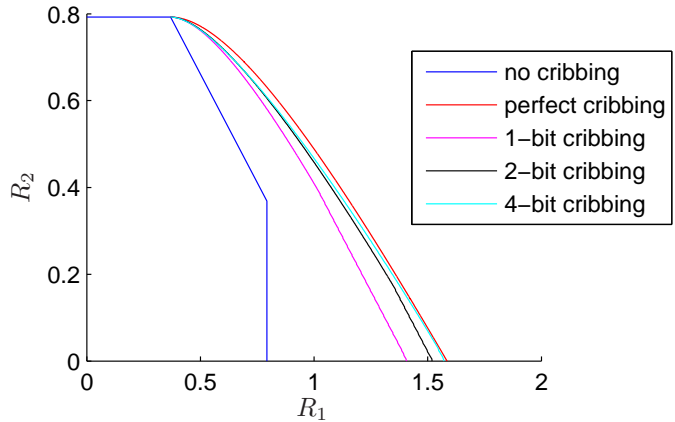}}

 \caption{Achievable regions of Gaussian MAC with a quantizer cribbing.}\label{f_mac_crib_Gauusian_graph}
 }\end{figure}
The input distribution is $P_{X_2|Z_1}(x_2|z_1)=\rho P_V(x_2)+(1-\rho)P_{X_1|Z_1}(x_2|z_1)$,
where $V\sim N(0,1)$ and is independent of $X_1$ and $Z$. Note that under these assumptions
$X_2\sim N(0,1)$ and therefore satisfies the power constraint. Fig.
\ref{f_mac_crib_Gauusian_graph} depicts the simple achievable scheme for different quantizers.
The blue line in Fig. \ref{f_mac_crib_Gauusian_graph} is the capacity region where there is no
cribbing, evaluated according to (\ref{e_gaussian_no_criibing}). The red line is the capacity
region where there is perfect cribbing, evaluated according to
(\ref{e_gaussian_perfect_cribbing}). The lines in between are achievable regions according to the
simple scheme we have described above. One can see that the main gain is already due to 1-bit
quantizer and that the difference between the achievable scheme with a 4-bit quantizer and the
capacity region where there is perfect cribbing is negligible.

\section{controlled cribbing \label{s_controlled_cribbing}}
Here we consider the case where the cribbing is controlled by an action which depends on
previously cribbed signals. In this study, only Encoder 2  cribs causally or strictly causally.
More precisely, at time $i$ there is a controller which takes action $a_{1,i}$
 and the cribbed signals from Encoder 1 to Encoder 2 at time $i$ is
$z_{1,i}=f(x_{1,i},a_{1,i})$  as shown in Fig. \ref{f_mac_crib_action}. The action at time $i$
depends on past cribbed observation, i.e., $a_{1,i}(z_1^{i-1})$ and the action is a limited
resource, namely, there is a restriction that $\frac{1}{n}\sum_{i=1}^n E[\Lambda(A_{1,i})]\leq
\Gamma,$ where $\Lambda(a_1)$ is a cost of taking action $a_1$.

\begin{figure}[h!]{
\psfrag{d}[][][0.9]{\;delay\ \ \ } \psfrag{B}[][][1]{Encoder 1} \psfrag{D}[][][1]{Encoder 2}
\psfrag{m1}[][][1]{$m_1 \ \ \ \ $} \psfrag{m2}[][][1]{$\in\{1,...,2^{nR_1}\}\ \ \ \ \ \ \ $}
\psfrag{m3}[][][1]{$m_2 \ \ \ \ $} \psfrag{m4}[][][1]{$\in\{1,...,2^{nR_2}\}\ \ \ \ \ \ \ $}
\psfrag{P}[][][1]{$\ \ P_{Y|X_1,X_2}$} \psfrag{x1}[][][1]{$\ \ \ \ \ \ \ \ \ \ \ \ X_{1,i}(m_1)$}
\psfrag{x2}[][][1]{\; \; \; \; \; $X_{2,i}(m_2,Z_1^{i})$} \psfrag{M}[][][1]{MAC}
\psfrag{s}[][][1]{$S$} \psfrag{Yi}[][][1]{$Y$} \psfrag{W}[][][1]{Decoder} \psfrag{t}[][][1]{$$}
\psfrag{G1}[][][0.9]{$z_{1,i}=g_1(a_{1,i}(z_1^{i-1}),x_{1,i})$}
\psfrag{G2}[][][0.9]{$z_{2,i}=g_2(a_{2,i}(z_2^{i-1}),x_{2,i})$} \psfrag{a}[][][1]{a}
\psfrag{b}[][][1]{b} \psfrag{c12}[][][1]{$ m_{12}\in$} \psfrag{c12b}[][][1]{$ \ \ \ \ \ \
\{1,...,2^{nC_{12}}\}$} \psfrag{Y}[][][1]{$\ \ \ \ \ \ \hat m_1(Y^n)$} \psfrag{Y2}[][][1]{$\ \ \
\ \hat m_2(Y^n)$} 
\centerline{\includegraphics[width=15cm]{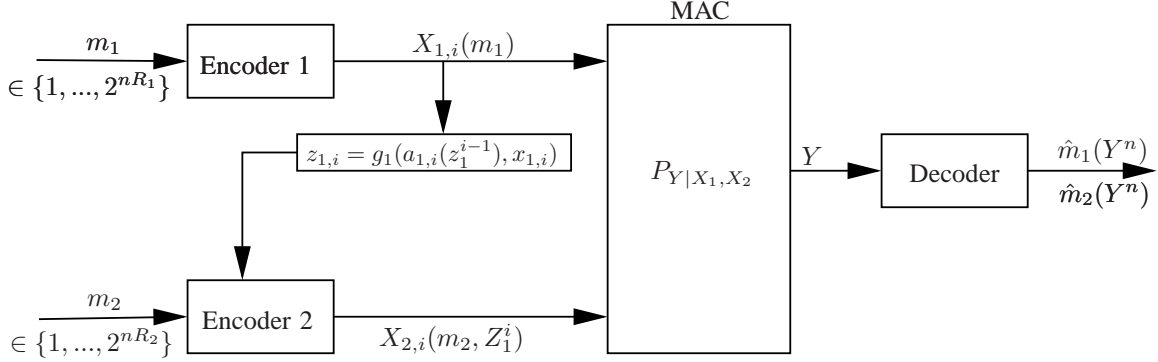}}
\caption{Partial  cribbing with actions. The action at time $i$ is $a_{1,i}$ and is determined by
previous cribbed observations i.e., $z_{1}^{i-1}$. The cribbed signal $z_{1,i}$ from Encoder 1 to
Encoder 2 is given by the deterministic function $z_{1,i}=g_1(a_{1,i},x_{1,i})$. There exists a
constraint on the actions of the form $\frac{1}{n}\sum_{i=1}^n E[\Lambda(A_{1,i})]\leq
\Gamma.$}\label{f_mac_crib_action}
 }\end{figure}

Let us now formally define a controlled code.
\begin{definition}\label{def_crib_act}
A $(2^{nR_1},2^{nR_2},n)$ {\it code} with controlled partial cribbing , as shown in Fig.
\ref{f_mac_crib_action}, consists at time $i$ of an encoding function at Encoder 1
\begin{eqnarray}
&\text{Case A, B,}&f_{1,i}:\{1,...,2^{nR_1}\} \mapsto \mathcal
\mathcal X_{1,i},
\end{eqnarray}
and an encoding function at Encoder 2 that changes according to the following case settings
\begin{eqnarray}
&\text{Case A}& f_{2,i}:\{1,...,2^{nR_2}\}\times  \mathcal Z_1^{i-1} \mapsto \mathcal
\mathcal X_{1,i},\nonumber \\
&\text{Case B}& f_{2,i}:\{1,...,2^{nR_2}\}\times  \mathcal Z_1^{i} \mapsto \mathcal \mathcal
X_{1,i},
\end{eqnarray}
and a controlled action
\begin{eqnarray}
g_{1,i}:  \mathcal Z_1^{i-1} \mapsto \mathcal
\mathcal A_{1,i},
\end{eqnarray}
 and a decoding function,
\begin{equation}
h:\mathcal Y^n \mapsto \{1,...,2^{nR_1}\} \times \{1,...,2^{nR_2}\}.
\end{equation}
The code needs to satisfy the constraint $\frac{1}{n}\sum_{i=1}^n E[\Lambda_1(A_{1,i})]\leq
\Gamma_1.$ The probability of error, achievable pair-rates and capacity region are defined in the
usual way for MAC as presented in Def. \ref{def_crib}.
\end{definition}

 Let us now define the following regions $\mathcal R_{A}^{a},\mathcal
R_{B}^{a}$, which are contained in $\mathbb{R}^2_{+}$, namely, contained in  the set of non
negative two dimensional real numbers.
\begin{equation}\label{e_region2_case_action}
\mathcal R_A^a=\left\{\begin{array}{l}R_1\leq H(Z_1|U,A_1)+I(X_1;Y|X_2,Z_1,U,A_1),\\
R_2\leq I(X_2;Y|X_1,U,A_1),\\
R_1+R_2\leq I(X_1,X_2;Y|U,A_1,Z_1)+H(Z_1|U,A_1),\\
R_1+R_2\leq I(X_1,X_2;Y), \text{ for }\\
 P(u,a_1)P(x_1,z_1|u,a_1)P(x_2|u,a_1)P(y|x_1,x_2) \text{ s.t. } E[\Lambda_1(A_{1})]\leq
\Gamma_1. \end{array} \right\}
\end{equation}
The region $\mathcal R_B^a$ is defined with the same set of inequalities as in
(\ref{e_region2_case_action}), but the joint distribution is of the form    \begin{equation}
 P(u,a_1)P(x_1,z_1|u,a_1)P(x_2|z_1,u,a_1)P(y|x_1,x_2)  \text{ s.t. } E[\Lambda(A_{1})]\leq
\Gamma.
\end{equation}


\begin{theorem}[Capacity region]\label{t_mac_one_crib_act}
The capacity regions of the MAC  with actions and with strictly-causal (Case A), and mixed causal
and strictly-causal (Case B), as described in Def. \ref{def_crib_act},
are $\mathcal R_A^a$, and  $\mathcal R_B^a$,  respectively. 
\end{theorem}

The proof is based on minor modification of the proof of the capacity region of the MAC with
partial cribbing presented in Theorem \ref{t_mac_one_crib}.
\begin{proof}

{\bf Achievability:} Consider the achievability proof of Theorem \ref{t_mac_one_crib} and just
replace $U_i$ by the pair $(U_i,A_{1,i})$. Note that the proof holds since at the end of block
$b$ the controller is able to decode $m'_{1,b}$.

{\bf Converse:} Consider the converse proof of Theorem \ref{t_mac_one_crib} and just replace
$U_i$. Note that $U_i\triangleq Z_1^{i-1}$. since $A_i$ is a function of $Z_1^{i-1}$ its also a
function of $U_i$ and by replacing $U_i$ by $U_i,A_{1,i}$ we obtain the converse proof
\end{proof}

\begin{example}[Deterministic Relay with actions]
Consider the case where only Encoder 1 has a message to transmit and Encoder 2 has no message of
its own to transmit, but helps to increase the rate of Encoder 1. Encoder 2, which plays the role
of a relay, takes an action $A_i$ that is a function of the observed signal up to time $i-1$,
i.e., $Z^{i-1}$. If $A_i=1$, then $Z_i=X_i$, and otherwise $Z_i$ is a constant. The cribbing
signal $Z_i$ is observed at Encoder 2 with a delay. There exists a constraint that
$\frac{1}{n}\sum_{i=1}^n E[A_i]\leq \Gamma$. In addition, Encoder 2 transmits  a signal $X_{2,i}$
through the channel at time $i$, where $X_{2,i}$ is a function of $Z^{i-1}$. The output channel
$Y$ is randomly chosen with equal probability to be either $X_1$ or $X_2$. This example is
illustrated in Fig. \ref{f_mac_crib_det_ex} and is a special case of the setting presented in
Fig. \ref{f_mac_crib_action}.

\begin{figure}[h!]{
\psfrag{B}[][][1]{Encoder 1} \psfrag{D}[][][1]{Encoder 2} \psfrag{D2}[][][1]{(Relay)\ \ \ }
\psfrag{m1}[][][1]{$m_1 \ \ \ \ $} \psfrag{m2}[][][1]{$\in\{1,...,2^{nR_1}\}\ \ \ \ \ \ \ $}
\psfrag{m3}[][][1]{$m_2 \ \ \ \ $}
\psfrag{m4}[][][1]{$\in\{1,...,2^{nR_2}\}\ \ \ \ \ \ \ $} \psfrag{P}[][][1]{$$}
\psfrag{x1}[][][1]{$\ \ \ \ \ \ \ \ \ \ \ \ X_{1,i}(m_1)$} \psfrag{x2}[][][1]{$X_{2,i}(Z^{i-1})$}
\psfrag{M}[][][1]{Channel}

 \psfrag{s}[][][1]{$S$} \psfrag{Yi}[][][1]{$Y$} \psfrag{W}[][][1]{Decoder}
\psfrag{t}[][][1]{$$} \psfrag{G1}[][][0.9]{$Z_{1,i}=g_1(X_{1,i})$}
\psfrag{G2}[][][0.9]{$Z_{2,i}=g_2(X_{2,i})$}

\psfrag{a}[][][1]{a} \psfrag{b}[][][1]{b} \psfrag{c12}[][][1]{$ m_{12}\in$}
\psfrag{c12b}[][][1]{$ \ \ \ \ \ \ \{1,...,2^{nC_{12}}\}$}

\psfrag{Y}[][][1]{$\ \ \ \ \ \ \hat m_1(Y^n)$} \psfrag{Y2}[][][1]{}
\psfrag{A0}[][][1]{$A=0$} \psfrag{A1}[][][1]{$A=1$} \psfrag{AZ}[][][1]{$A_i(Z^{i-1})$}
\psfrag{Z}[][][1]{$Z_i\ $} \psfrag{X1}[][][1]{$X_1$} \psfrag{X2}[][][1]{$X_2$}

\centerline{\includegraphics[width=14cm]{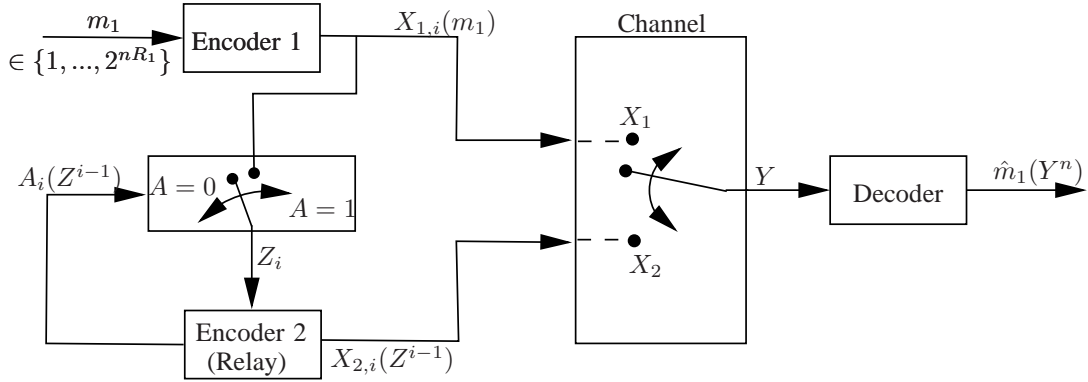}}


\caption{An example of deterministic cribbing with actions. The relay (Encoder 2) take an action
$A_i$ at time $i$ that depends on previous cribbing, i.s., $Z^{i-1}$. The cribbing signal $Z_i$
equals to $X_{1,i}$ if $A_i=1$ and is constant otherwise. The cribbing is a limited resource
hence there exists a constraint that on the portion of time that Encoder 2 can crib the signal
from Encoder 1, namely, $\frac{1}{n}\sum_{i=1}^n E[A_i]\leq \Gamma$. The output channel $Y$ is
randomly chosen with equal probability to be either $X_1$ or $X_2$}\label{f_mac_crib_det_ex}
 }\end{figure}

The next lemma establishes the capacity region of a deterministic relay with actions which is a
special case of the cribbing with actions.
\begin{lemma}
The capacity region of partial deterministic cribbing with actions where only Encoder 1 sends a
message, i.e., $R_2=0$ and there exists a delay in the cribbing (Case A) is
\begin{equation}\label{e_action1}
R_1=\max_{P_{X_1,X_2,A}: E[c(A)]\leq \Gamma}\min\{H(Z|X_2,A)+I(X_1;Y|X_2,Z_1,A), I(Y;X_1,X_2) \}.
\end{equation}
 If there is no delay in the cribbing (Case B), i.e., $X_{2,i}(Z^i)$, then
\begin{equation}
R_1=\max_{P_{U,X_1,A}P_{X_2|U,Z,A}: E[c(A)]\leq \Gamma}\min\{H(Z|U,A)+I(X_1;Y|X_2,Z,U,A),
I(Y;X_1,X_2) \}.
\end{equation}
\end{lemma}

\begin{proof}
Since $R_2=0$ follows from (\ref{e_region2_case_action}) that
\begin{eqnarray}\label{e_action_delay}
R_1&\leq&\max_{\mathcal P} \min\{ H(Z|U,A)+I(X_1;Y|X_2,Z,U,A), I(X_1,X_2;Y)\}.
\end{eqnarray}
For the case where there is a delay in the cribbing (case A) the set of joint distributions
$\mathcal P$ is of the form $P(u,a,x_1)P(x_2|u,a)P(y|x_1,x_2)$ and $Z$ is a function of $A$ and
$X$. Using mathematical manipulation on the first term in the minimum in (\ref{e_action_delay})
we obtain
\begin{eqnarray}\label{e_action_in}
R_1&\stackrel{(a)}{\leq}& H(Z|U,A,X_2)+I(X_1;Y|X_2,Z,U,A)\nonumber\\
&\stackrel{(b)}{\leq}& H(Z|A,X_2)+I(X_1;Y|X_2,Z,A),
\end{eqnarray}
where step (a) follows from the Markov chain $X_2-(U,A)-X_1-Z$ and step (b) from the fact that
conditioning reduces entropy and the Markov chain $Y-(X_1,X_2,Z_A)-U$. By choosing $U=X_2$ the
first term of (\ref{e_action_delay}) become the upper bound in (\ref{e_action_in}); hence
(\ref{e_action1}) is the capacity region.

In the case that there is no delay in the cribbing the capacity region is simply
(\ref{e_action_delay}) where the set of joint distribution $\mathcal P$ is of the form
$P(u,a,x_1)P(x_2|u,a,z)P(y|x_1,x_2)$ and $z$ is a deterministic function of $a$ and $x$.
\end{proof}

For the case of delay in the cribbing, the action $A_i$ can be seen as part of the output signal
from Encoder 2 to the channel, and indeed by replacing $X_2$ in
(\ref{e_semi_deterministic_with_delay}) with $(X_2,A)$, we obtain (\ref{e_action1}). However, in
the case of no delay in the cribbing i.e., $X_2(Z^{i})$, the replacement of $X_2$ is not possible
since the action must have a delay i.e., $A_i(Z^{i-1})$.

For obtaining a numerical solution when there is a delay in the cribbing, namely, evaluating
(\ref{e_action1}) for the example in Fig. \ref{f_mac_crib_det_ex} we can assume without loss of
optimality that
\begin{eqnarray}
\Pr(A=1)&=&\Gamma, \nonumber \\
\Pr(X_1=X_2|A=0)&=&\alpha_0, \nonumber \\
\Pr(X_1=X_2|A=1)&=&\alpha_1.
\end{eqnarray}
The reason one can assume that $\Pr(A=1)=\Gamma$ is because if this is not the case, and one has
a code where the portion of $\Pr(A=1)$ is smaller than $\Gamma,$ then one can add actions $A=1$
for
 some portion of time without decreasing the performance of the code. Furthermore, since the
 channel is symmetric  with respect to $0$ and $1$ (by exchanging 0 and 1 for the
 inputs to the channels the performance of the code remains the same)  only  the
 probability $\Pr(X_1=X_2)$ is important. Furthermore, from the same reasons  one can also assume that $P(x_1)$
and $P(x_2)$
 are Bernoulli$(\frac{1}{2})$ without loss of optimality. Now we shall compute the terms in (\ref{e_action1})
 \begin{eqnarray}
 I(Y;X_1,X_2)&=&H(Y)-H(Y|X_1,X_2)\nonumber \\
  &\stackrel{}{=}& 1- \Gamma+\alpha_1\Gamma-(1-\Gamma)(1-\alpha_0)\nonumber \\
   &=&\alpha_1\Gamma+\alpha_0(1-\Gamma),\\
\nonumber \\
H(Z|X_2,A)&=& \Gamma H_b(\alpha_1),\\
\nonumber \\
 I(X_1;Y|X_2,Z,A)&=&H(Y|X_2,Z,A)-H(Y|X_1,X_2,A)\nonumber \\
 &\stackrel{(a)}{=}&\Gamma(1-\alpha_1)+(1-\Gamma)H_b\left(\frac{1+\alpha_0}{2}\right)-\Gamma(1-\alpha_1)-(1-\Gamma)(1-\alpha_0)\nonumber \\
  &=& (1- \Gamma)\left(H_b\left(\frac{1+\alpha_0}{2}\right)+\alpha_0-1 \right),\label{e_compute_ex}
 \end{eqnarray}
where step (a) in (\ref{e_compute_ex}) is due to the fact that
$\Pr(Y=X_2|X_2,a=0)=\alpha_0+\frac{1-\alpha_0}{2}$ and therefore
$H(Y|X_2,Z,A)=\Gamma(1-\alpha_1)+(1-\Gamma)H_b(\frac{1+\alpha_0}{2})$ where $H_b(p)$ is the
binary entropy, i.e., $-p\log p -(1-p) \log (1-p)$ for $0\leq p\leq 1$. Hence the capacity of the
setting in Fig. \ref{f_mac_crib_det_ex} as a function on the constrain of the action $\Gamma$ is
\begin{equation}\label{C_Gamma}
C(\Gamma)=\max_{0\leq\alpha_0,\alpha_1\leq 1} \min(\Gamma H_b(\alpha_1)+(1-
\Gamma)\left(H_b\left(\frac{1+\alpha_0}{2}\right)+\alpha_0-1 \right),
\alpha_1\Gamma+\alpha_0(1-\Gamma)).
\end{equation}
\begin{figure}[h!]{
\psfrag{Z}[][][1]{$H_b(\frac{1}{5})-\frac{2}{5} \to  \ \ \ \ \ \ \ \ \ $} \psfrag{max}[][][1]{$\
\ \ \ \ \ \ \ \ \ \ \ \ \ \ \ \ \  \ \ \ \ \ \ \leftarrow \max_\alpha \min(\alpha,H_b(\alpha))$}
\psfrag{gamma}[][][1]{$\Gamma$} \psfrag{R}[][][1]{$C(\Gamma)$}
\centerline{\includegraphics[width=8cm]{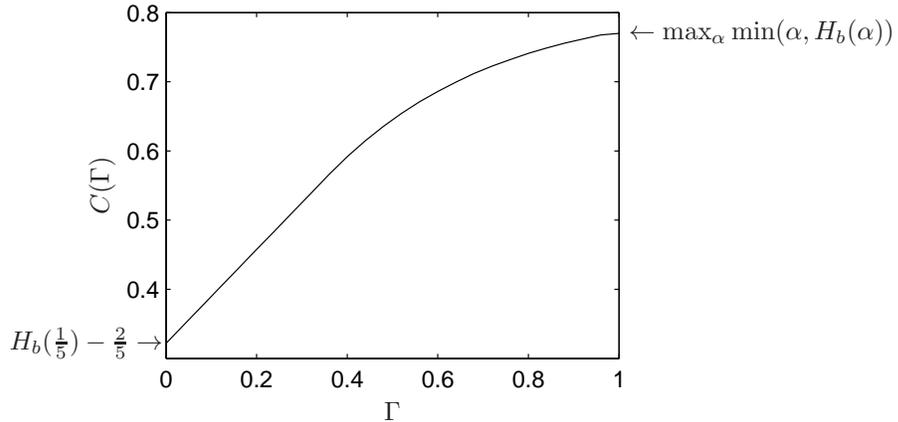}} 

\caption{Capacity of setting in  Fig. \ref{f_mac_crib_det_ex} as a function of the action
constraint $\Gamma$. For the case where $\Gamma=0$ the capacity can be solved analytically since
it is the capacity of the $Z$ channel. The capacity where  $\Gamma=1$ is the simple expression
$\max_{\alpha_1} \min({\alpha_1},H_b({\alpha_1}))$ which can be solved numerically by solving
${\alpha}=H_b({\alpha}).$}\label{f_ex_deter}
 }\end{figure}
The capacity $C(\Gamma)$ is depicted in Fig. \ref{f_ex_deter} and can be found simply using a
grid-search on $0\leq\alpha_0,\alpha_1\leq 1$ or by convex optimization tools. In the case that
$\Gamma=0$, $X_{2,i}$ is independent of the message $m_1$ and therefore we obtain that at any
time $i$ the channel from Encoder 1 to the output behaves as a $Z-$channel if $X_{2,i}=0$ and as
an $S$ channel if $X_{2,i}=1$ and the capacity of those two channels are
$H_b(\frac{1}{5})-\frac{2}{5}$, and therefore $C(0)=H_b(\frac{1}{5})-\frac{2}{5}.$ For the case
that $\Gamma=1$ we obtain from (\ref{C_Gamma}) that $C(1)=\max_{\alpha_1}
\min({\alpha_1},H_b({\alpha_1}))$. The $\alpha$ that maximizes the expression of $C(1)$ is the
one that solves the equation ${\alpha_1}=H_b({\alpha_1}).$
\end{example}

\section{Conclusions and further research directions\label{s_conclusion}}
We have considered the problem of MACs with partial cribbing encoders, namely, in a two encoder
MAC the observed cribbed signal is a deterministic function of the other encoder output. We have
characterized the capacity region for the two cases where the partial cribbing is causal or
strictly causal. Rate splitting  is the main additional technique used in the achievability proof
over the techniques used for perfect cribbing. The extension of perfect cribbing to partial
cribbing resemble to the extension of the decode-and-forward technique for the relay to the
partial-decode-and-forward technique
 \cite{ElGammalKim10LectureNotes}. The method we used for
partial cribbing may be also used for noisy cribbing, although in general the capacity region of
noisy cribbing is an open question. Another question that has not been solved yet is the non
causal partial cribbing. For the perfect cribbing case Willems \cite{Willems85_cribbing_encoders}
solved the noncausal case simply by showing that causal and non-causal perfect cribbing results
in the same capacity region.

Solving the partial cribbing setting allowed us to solve an action dependent cribbing problem. In
this paper we considered the case where the action is only a function of the previously observed
cribbing. However, the case in Fig. \ref{f_mac_crib_action} where the action is a function of the
previously observed cribbing and the message of the cribbing encoder, i.e.,
$a_{1,i}(z_1^{i-1},m_1)$ is yet to be solved. Issues of this nature may be raised in the sphere
of  cognitive communication systems where sensing other users' signals is a resource with a cost.

\section*{Acknowledgement}
The authors are grateful to Tsachy Weissman for very helpful discussions and for supporting the
work through NSF grants CCF-1049413 and 4101-38047, and to Yossef Steinberg for pointing out to
us the motivation of partial cribbing for the Gaussian case and for additional helpful
discussions. The work of Haim Permuter was partially supported by Marie-Curie fellowship and the
work of Himanshu Asnani  by The Scott A. and Geraldine D. Macomber Stanford Graduate Fellowship
Fund.

\bibliographystyle{unsrt}
\bibliographystyle{IEEEtran}

\end{document}